\pdfoutput=1
\documentclass{eptcs}

\usepackage{amsmath,amsthm,amsfonts,cleveref,braket,xcolor,enumitem}

\usepackage[firstinits=true,maxcitenames=5,backend=bibtex,sorting=none]{biblatex}
\DeclareFieldFormat[article]{volume}{\textbf{#1}}
\DeclareFieldFormat{pages}{#1}
\renewbibmacro{in:}{%
  \ifentrytype{article}{}{\printtext{\bibstring{in}\intitlepunct}}}
\defbibheading{bibliography}{\section*{#1}}

\newcommand{\qp}{\mathbb{P}}
\newcommand{\op}{\mathrm{Pr}}
\newcommand{\pvm}{\mathcal{M}}
\newcommand{\abl}[2][\pvm]{\qp (#2|\psi, #1, \phi)}
\newcommand{\qabl}[1]{\qp (#1|\psi,\phi)}
\newcommand{\abs}[1]{\left\lvert{#1}\right\rvert}
\newcommand{\prj}[1]{\ket{#1}\bra{#1}}
\newcommand{\chn}{\mathcal{C}}
\newcommand{\pvmc}{\mathcal{E}}
\newcommand{\pset}{\mathcal{P}}
\newtheorem{lem}{Lemma}
\newtheorem{thm}{Theorem}
\newtheorem{prop}{Proposition}

\crefname{thm}{Theorem}{Theorems}
\Crefname{thm}{Theorem}{Theorems}
\crefname{cond}{condition}{conditions}
\DeclareMathOperator{\tr}{Tr}

\begin{document}

\title{Logical pre- and post-selection paradoxes are proofs of
  contextuality} 

\author{Matthew F. Pusey \qquad\qquad Matthew S. Leifer
  \institute{Perimeter Institute for Theoretical Physics, Waterloo ON,
    Canada\thanks{Research at Perimeter Institute is supported in part
      by the Government of Canada through NSERC and by the Province of
      Ontario through MRI\@. ML is supported by the Foundational
      Questions Institute (FQXi). Thanks to Joshua Combes, Chris
      Ferrie, Bob Griffiths, Owen Maroney and Rob Spekkens for discussions.}}  \email{\qquad
    m@physics.org \qquad\qquad\qquad matt@mattleifer.info} }

\newcommand{\titlerunning}{Logical pre- and post-selection paradoxes
  are proofs of contextuality}

\newcommand{\authorrunning}{M. F. Pusey and M. S. Leifer}

\newcommand{\event}{QPL 2015}

\maketitle

\begin{abstract}
  If a quantum system is prepared and later post-selected in certain
  states, ``paradoxical'' predictions for intermediate measurements
  can be obtained.  This is the case both when the intermediate
  measurement is strong, i.e.\ a projective measurement with
  L{\"u}ders-von Neumann update rule, or with weak measurements where
  they show up in anomalous weak values.  Leifer and Spekkens
  [Phys. Rev. Lett. \textbf{95}, 200405] identified a striking class
  of such paradoxes, known as \emph{logical} pre- and post-selection
  paradoxes, and showed that they are indirectly connected with
  contextuality. By analysing the measurement-disturbance required in
  models of these phenomena, we find that the strong measurement
  version of logical pre- and post-selection paradoxes actually
  constitute a direct manifestation of quantum contextuality. The
  proof hinges on under-appreciated features of the paradoxes.  In
  particular, we show by example that it is not possible to prove
  contextuality without L{\"u}ders-von Neumann updates for the
  intermediate measurements, nonorthogonal pre- and post-selection,
  and $0$/$1$ probabilities for the intermediate measurements.  Since
  one of us has recently shown that anomalous weak values are also a
  direct manifestation of contextuality
  [Phys. Rev. Lett. \textbf{113}, 200401], we now know that this is
  true for both realizations of logical pre- and post-selection
  paradoxes. 
\end{abstract}

\section{Introduction}

Can a ball be in two separate boxes at once, and does the answer to
this question depend in any meaningful way upon quantum mechanics?
Issues such as these have been raised by a series of colourfully
described thought experiments involving pre- and post-selected quantum
systems.

Suppose a quantum system is prepared in state $\ket{\psi}$, subjected
to an intermediate projective measurement $\pvm = \{P_j\}$ with
L{\"u}ders-von Neumann update rule\footnote{This is the traditional
  ``projection postulate'' where, upon obtaining the result $P_j$, the
  state of the system is updated to $P_j
  \ket{\psi}/\bra{\psi}P_j\ket{\psi}$.}, followed by a final
projective measurement that includes the projector onto $\ket{\phi}$
as one of its outcomes. Assuming that no other evolution occurs, the
joint probability for obtaining the outcome $P_j$ and passing the
post-selection is
\begin{equation}
  \qp(P_j,\phi|\psi,\pvm) = \abs{\Braket{\phi|P_j|\psi}}^2,
\end{equation}
and the marginal probability for passing the post-selection
is then
\begin{equation}
  \qp(\phi|\psi,\pvm) = \sum_j \qp(P_j,\phi|\psi,\pvm) = \sum_j
  \abs{\Braket{\phi|P_j|\psi}}^2. 
\end{equation}
From this, we can calculate the probabilities for the intermediate
measurement conditioned on both the pre- and post-selection as
\begin{equation}
  \label{ablrule}
  \abl{P_j} = \frac{\qp(P_j,\phi|\psi,\pvm)}{\qp(\phi|\psi,\pvm)} =
  \frac{\abs{\Braket{\phi|P_j|\psi}}^2}{\sum_k
    \abs{\Braket{\phi|P_k|\psi}}^2}, 
\end{equation}
which is known as the ``ABL rule'' \cite{abl}.

Various choices of $\ket{\psi}$, $\pvm$ and $\ket{\phi}$ have been
shown to give counter-intuitive results, for example the ``three-box
paradox'' \cite{given}, ``quantum cheshire cats'' \cite{chesire}
and recently the ``quantum pigeonhole principle'' \cite{pigeonhole}.

For example, the three-box paradox involves a state space spanned by
$\{\ket{1}, \ket{2}, \ket{3}\}$ representing a ball in box $1$, $2$,
or $3$ respectively. Consider a pre-selection $\ket{\psi} \propto
\ket{1} + \ket{2} + \ket{3}$ and a post-selection $\ket{\phi} \propto
\ket{1} + \ket{2} - \ket{3}$. If we ``look in box 1'', $\pvm = \{
\prj{1}, \prj{2}+\prj{3}\}$, then whenever the post-selection succeeds
we will have found the ball, $\abl{\prj{1}} = 1$. But if instead we
``look in box 2'', $\pvm' = \{ \prj{1}+\prj{3}, \prj{2}\}$, we also
have $\abl[\pvm']{\prj{2}} = 1$. Hence the ball is in both boxes.

Or is it? In addition to general concerns about the interpretation of
ABL probabilities for unperformed measurements (e.g. \cite{kastner}),
ontological models (without balls that are in more than one box)
reproducing various aspects of the paradox have been proposed
\cite{kirk,lsmodel,maroney} and criticised \cite{revisited}. The basic
idea of such models is that the intermediate measurement can disturb
the system, thus allowing the success of the post-selection to depend
on which measurement was performed.

We believe the central question is this: does a given pre- and
post-selection (PPS) phenomenon have a compelling classical
explanation? And we believe the best way to make this question precise
is: does the phenomena admit a non-contextual ontological model?

The most well-known obstruction to non-contextual models of quantum
theory is the Kochen-Specker theorem \cite{ks}. The question of
whether certain PPS paradoxes constitute proofs of the Kochen-Specker
theorem has been discussed, and answered in the negative
\cite{bub,bubcomment,lsctx,lsmodel}. Again the crucial issue is a
non-contextual assignment of values to the intermediate measurement
may appear contextual under post-selection, due to measurement
disturbance.

Nevertheless it was found in \cite{lsctx} that certain PPS paradoxes,
which were dubbed \emph{logical} PPS paradoxes, may be converted into
proofs of the Kochen-Specker theorem by considering a standard
``prepare and measure'' experiment (without post-selection) in which
the intermediate measurements along with two additional measurements
(based on what were the pre- and post-selection) are all considered as
counterfactual alternatives. This leaves the status of the logical PPS
paradox itself somewhat unclear.

Here we show that, by analysing the possible disturbance due to the
intermediate measurement in a non-contextual model, the paradoxes, in
their original form, are in fact proofs of contextuality in the sense
of \cite{cntx}, which generalises Kochen-Specker non-contextuality to
include preparations and Positive Operator Valued Measures
(POVMs). Whilst previous discussions have centred on the
\emph{existence} of measurement disturbance, it turns out that the
\emph{amount} of disturbance permitted by non-contextuality (whilst
non-zero) is insufficient to dissolve the paradox. Hence we will show

\begin{thm}\label{mainthm}
  Every logical PPS paradox is a proof of contextuality.
\end{thm}

\section{Logical pre- and post-selection paradoxes}

From now on, we shall only consider ABL probabilities of the form
$\abl[\{P,I-P\}]{P}$ where the intermediate measurement has two
outcomes and $I$ is the identity operator. Since the projective
measurement $\{P,I-P\}$ is thus uniquely determined by $P$, we shall
abbreviate $\abl[\{P,I-P\}]{P}$ to $\qabl{P}$ and
$\qp(\phi|\psi,\{P,I-P\})$ to $\qp(\phi|\psi)$.

Our definition of a logical PPS paradox is based on
\cite{lsctx,lsmodel}.  Consider a Hilbert space, a choice of
pre-selection $\ket{\psi}$ and post-selection $\ket{\phi}$, and a
(finite) set of projectors $\pset$ that is closed under complements,
i.e.\ if $P \in \pset$ then $I-P \in \pset$.  Suppose further that the
ABL probabilities $\qabl{P}$ are either $0$ or $1$ for every $P \in
\pset$ (which is what leads to the terminology ``logical'').

Now consider the partial boolean algebra generated by $\pset$,
i.e.\ the smallest set of projectors $\pset'$ that contains
$\pset$ and satisfies
\begin{itemize}
  \item If $P \in \pset'$ then $I - P \in \pset'$.
  \item If $P, Q \in \pset'$ and $PQ = QP$ then $PQ \in \pset'$.
\end{itemize}
If we think of projectors as representing propositions, then these
conditions ensure that we can take complements and conjunctions of
compatible propositions.\footnote{Since $P + Q - PQ = I - (I-P)(I-Q)$
  we can also take disjunctions. We thank a referee for
this simplification.}

Finally, suppose that we try to extend the probability function $f(P)
= \qabl{P}$ from $\pset$ to $\pset'$ such that the following
\emph{algebraic conditions} are satisfied\footnote{
  \cite{lsctx,lsmodel} gave an additional condition $f(P'Q') \leq f(P')$
  when $P'Q' = Q'P'$. But this follows from \cref{ac3} with $P =
  P'Q'$ and $Q = P'(I-Q')$, and then using \cref{ac1,ac2}.
  Also note that \cref{ac2,ac3} give $f(I-P) = 1 - f(P)$.}
\begin{enumerate}[label=(\roman*)]
\item For all $P \in \pset'$, $0 \leq f(P) \leq 1$. \label[cond]{ac1}
\item $f(I) = 1, f(0) = 0$. \label[cond]{ac2}
\item For all $P,Q \in \pset'$ such that $PQ = QP$, $f(P + Q - PQ) =
  f(P) + f(Q) - f(PQ)$. \label[cond]{ac3}
\end{enumerate}
If it is not possible to do this then we say that the ABL predictions
for $\pset$ form a \emph{logical PPS paradox}.

For example in the three-box paradox we have $f(\prj1) =
\qabl{\prj{1}} = 1$ and $f(\prj{2}) = \qabl{\prj{2}} = 1$. Applying
\cref{ac3} gives $f(\prj{1} + \prj{2}) = 2$ in violation of
\cref{ac1}. Other examples can be found in \cite{prodrule, prodrule2,
  peculiar, pigeonhole, cabello}.

The following simple proposition will be useful later.
\begin{prop}
  \label{orthoprop}
  In a logical PPS paradox, the pre-selection $\ket{\psi}$ and
  post-selection $\ket{\phi}$ are necessarily nonorthogonal.
\end{prop}
\begin{proof}
  According to the definition of a logical PPS paradox, the ABL
  probabilities $\qabl{P}$ assigned to the projectors $P \in \pset$
  must be $0$ or $1$.  However, if $\ket{\psi}$ and $\ket{\phi}$ are
  orthogonal, then no $0$/$1$ probability assignments are possible.
  To see this, suppose that $\qabl{P} = 1$.  Then, from
  \cref{ablrule},
  \begin{equation}
    \abs{\braket{\phi|P|\psi}}^2 = \abs{\braket{\phi|P|\psi}}^2 +
    \abs{\braket{\phi|(I-P)|\psi}}^2, 
  \end{equation}
  which implies $\braket{\phi|(I-P)|\psi} = 0$.  This implies that
  \begin{align}
    \braket{\phi|P|\psi} & = \braket{\phi|P|\psi} +
    \braket{\phi|(I-P)|\psi} \\
    & = \braket{\phi|(P+I-P)|\psi} \\
    & = \braket{\phi|\psi},
  \end{align}
  which is also zero if the pre- and post-selection are orthogonal.
  This means the post-selection never occurs, so there are no ABL rule probabilities and hence no paradox.  A similar argument shows that the same is true for
  $\qabl{P} = 0$.
\end{proof}

There is an obvious similarity between logical PPS paradoxes and the
Kochen-Specker theorem.  Briefly, a Kochen-Specker noncontextual model
can be thought of as an assignment of values $v(P) \in \{0,1\}$ to
projection operators such that the algebraic conditions are satisfied
with $f(P) = v(P)$.  The Kochen-Specker theorem shows that such an
assignment is not possible in general.  However, in the Kochen-Specker
scenario, the value assignments represent the predictions of a
hypothetical outcome-deterministic ontological model (or hidden
variable theory if you prefer archaic terminology), which are supposed
to reproduce the quantum predictions in an ordinary
``prepare-and-measure'' experiment (i.e.\ with no post-selection) when
we average them over a probability measure.  In contrast, the ABL
probabilities represent the \emph{quantum} predictions with both pre-
and post-selection, and there is a possibility that the intermediate
measurements may disturb the state of the system, changing the
probability of whether the post-selection is successful.  Thus, no
direct inference from logical PPS paradoxes to Kochen-Specker
contextuality is possible.  To establish contextuality from logical
PPS paradoxes, we shall therefore have to look deeper, employing the
more general definition of non-contextuality from \cite{cntx}, which
allows us to place constraints on the amount of disturbance that can
occur in a non-contextual model.

\section{Non-contextual ontological models}

By a proof of contextuality, we mean a proof of the impossibility of a
non-contextual ontological model, as defined in \cite{cntx}. We will
need two facets of the assumption of non-contextuality:
\emph{measurement non-contextuality} and \emph{outcome determinism for
  sharp measurements}.

Briefly, a non-contextual ontological model associates a quantum
system with a measurable space $(\Lambda, \Sigma)$ where $\Lambda$ is
the set of ``ontic states'' and $\Sigma$ is a $\sigma$-algebra, a
preparation with a measure $\mu$ on $(\Lambda,\Sigma)$, and POVM
elements $E$ with conditional probabilities $\op (E|\lambda)$, such that
$\sum_{E \in \pvm} \op (E|\lambda) = 1$ for every POVM $\pvm$ and every
$\lambda \in \Lambda$. Upon marginalising over the ontic states, the
model is required reproduce the quantum probabilities:
\begin{equation}
  \int_{\Lambda} \op (E|\lambda)d \mu(\lambda) =
  \braket{\psi|E|\psi}.  
\end{equation}
where $\ket{\psi}$ is the prepared quantum state.

The assumption of measurement non-contextuality has already been made,
namely that the conditional probability of obtaining outcome $E$,
$\op(E|\lambda)$, depends only on the POVM element $E$, and not on
the other POVM elements in the POVM being measured nor on how it is
measured (e.g\ which other POVM it was obtained from by
coarse-graining).

Measurement non-contextuality has two consequences that we shall make
use of in the proof of \cref{mainthm}.  Firstly, if a POVM $\{E_j\}$
can be obtained by coarse-graining a POVM $\{E_{jk}\}$, i.e.\ $E_j =
\sum_k E_{jk}$, then
\begin{equation}
  \label{coarse}
  \op (E_j|\lambda) = \sum_k \op (E_{jk}|\lambda).
\end{equation}
This is because one method of measuring the POVM $\{E_j\}$ is to
measure the POVM $\{E_{jk}\}$ and then subsequenty marginalise over
$k$, and all methods of measuring a POVM must give the same
probabilities $\op (E_j|\lambda)$.  Secondly, for similar reasons, if
a POVM $\{E_j\}$ is a mixture of two POVMs $\{E_j'\}$ and $\{E_j''\}$,
i.e.\ $E_j = qE'_j + (1-q)E''_j$ for some $0 \leq q \leq 1$, then
\begin{equation}
  \label{mix}
  \op (E_j|\lambda) = q\op (E'_j|\lambda) + (1-q)\op (E''_j|\lambda).
\end{equation}
This is because one method of measuring $\{E_j\}$ is to flip a biased
coin with probability $q$ of coming up heads, measure $\{E_j'\}$ if
heads is obtained or $\{E_j''\}$ if tails is obtained, and then
subsequently only recording the outcome $j$.

The assumption of outcome determinism for sharp measurements is that
$\op (E|\lambda) \in \{0,1\}$ whenever $E$ is a projector. Rather than
being assumed, it can be derived from a version of non-contextuality
for preparations, together with some basic facts about projective
measurements in quantum theory. For the details of this argument see
\cite{cntx,robrant}.

It is straightforward to check that these assumptions imply that that
the assignments $f(P) = \op (P|\lambda)$ have to satisfy the
algebraic conditions. However, the possibility of measurement
disturbance blocks a direct inference from the observed pre- and
post-selected probabilities to the probabilities conditioned only on
the pre-selection of an ontic state \cite{lsmodel}. In order to prove \cref{mainthm},
we therefore need to understand the type of disturbance caused by a
projective measurement. It will turn out to be important that the
channel induced by ignoring the outcome of such a measurement can also
be implemented in a way that involves, with non-zero probability,
doing nothing.

\begin{lem}\label{idpart}
  For each projective measurement $\{P_j\}$ there exists a non-zero probability $q$ and a quantum channel (i.e. a completely-positive trace-preserving map) $\chn$ such that
  \begin{equation}
    \sum_j P_j \rho P_j = q\rho + (1-q)\chn(\rho) \qquad \forall \rho.\label{idparteq}
  \end{equation}
\end{lem}
\begin{proof}
  Suppose $j$ runs from $1$ to $n$. Let $X = \{1,-1\}^n$, i.e.\ the
  set of all strings $x = (x_1,x_2,\ldots,x_n)$ where $x_j = \pm
  1$. For $x \in X$ define
  \begin{equation}
    U_x = \sum_{j=1}^n x_j P_j
  \end{equation}
  which is unitary since $U_x^\dagger U_x = \sum_{j,k = 1}^n x_j x_k
  P_j^\dagger P_k = \sum_{j=1}^n {x_j}^2 P_j = \sum_{j=1}^n P_j = I$,
  where we have used that $\{P_j\}$ is a set of orthogonal projectors
  summing to the identity.

  Now consider $\sum_{x \in X} x_j x_k$. If $j=k$ then this is
  $\sum_{x \in X} 1 = 2^n$. Otherwise, the number of strings with
  $(x_j,x_k) = (1,1)$, $(x_j,x_k) = (-1,-1)$, $(x_j,x_k) = (1,-1)$,
  and $(x_j,x_k) = (-1,1)$ are all equal, with the first two sets
  contributing $1$ to the sum and the second to contributing
  $-1$. Hence $\sum_{x \in X} x_j x_k = 2^n \delta_{jk}$, and so
  \begin{equation}
    \frac{1}{2^n}\sum_{x \in X} U_x \rho U_x^\dagger = \frac{1}{2^n}
    \sum_{j,k = 1}^n \sum_{x \in X} x_j x_k P_j \rho P_k = \sum_{j,k =
      1}^n \delta_{jk} P_j \rho P_k \\
    = \sum_j P_j \rho P_j. 
  \end{equation}.

  Since $U_{\pm(1,\dotsc,1)} = \pm I$ we have \cref{idparteq} with $q
  = 2^{1-n}$ and $\chn(\rho) \propto \sum_{x \neq \pm(1,\dotsc,1)} U_x
  \rho U_x^\dagger$.
\end{proof}

Let us see the implication of this for the disturbance.

\begin{lem}
  \label{mdlem}
  Let $\{E_k\}$ be a POVM, let $\{P_j\}$ be a projective measurement,
  and let $\pvmc$ be the channel $\pvmc(\rho) = \sum_j P_j \rho P_j$,
  corresponding to performing the measurement $\{P_j\}$ and not
  recording the outcome.  In a measurement non-contextual model, if
  $\lambda$ makes some outcome of $\{E_k\}$ possible, i.e. $\op
  (E_k|\lambda) > 0$ for some $k$, then $\op
  (\pvmc^{\dagger}(E_k)|\lambda) > 0$, where $\pvmc^{\dagger}$ is the
  adjoint channel to $\pvmc$, i.e.\ $\lambda$ also makes the $k$th
  outcome possible in the measurement procedure consisting of
  performing $\{P_j\}$ and not recording the outcome, followed by
  performing $\{E_k\}$.
\end{lem}

\begin{proof} 
  The effect of performing the measurement $\{P_j\}$ and not recording
  the outcome is given by the channel
  \begin{equation}
    \pvmc(\rho) = \sum_j P_j \rho P_j.
  \end{equation}
  If we apply this channel to a state $\rho$, then measure
  $\{E_k\}$, the probabilities are given by $\tr(E_k\pvmc(\rho)) =
  \tr(\pvmc^\dagger(E_k)\rho)$ where $\pvmc^\dagger$ is the adjoint
  channel to $\pvmc$ (for our channel $\pvmc^\dagger = \pvmc$ but it
  will useful to keep the conceptual distinction). So we can consider
  the overall procedure as a measurement of the POVM $\{
  \pvmc^\dagger(E_k) \}$.

  Now consider another procedure. With probability $q$ we simply
  measure $\{E_k\}$, whereas with probability $1-q$ we measure
  $\{\chn^\dagger(E_k) \}$, where $q > 0$ and $\chn$ are from
  \cref{idpart}. By \cref{mix}, in the ontological model this will
  correspond to $q \op (E_k|\lambda) + (1-q) \op
  \left(\chn^\dagger(E_k)|\lambda\right)$.

  But by \cref{idpart} we have $\pvmc^\dagger(E_k) = qE_k +
  (1-q)\chn^\dagger(E_k)$, and so these two procedures correspond to
  the same POVM\@. By measurement non-contextuality, we therefore have
  \begin{equation}
    \op \left( \pvmc^\dagger(E_k)|\lambda \right) = q \op (E_k|\lambda) +
    (1-q) \op \left(\chn^\dagger(E_k)|\lambda\right) \geq q
    \op (E_k|\lambda), 
  \end{equation}
  so that $\op (E_k|\lambda) > 0$ implies $\op \left(
    \pvmc^\dagger(E_k) |\lambda \right) > 0$.
\end{proof}

In other words, the measurement-disturbance of a projective
measurement cannot make an outcome of a following measurement go from
being possible to impossible.\footnote{It is worth noting that the proof features an additional measurement $\chn^\dagger(E_k)$, which does not appear in definition of the paradox. But this is merely a device for getting an operational handle on the measurement-disturbance, and no particular facts about $\chn^\dagger(E_k)$ or its representation in the ontological model (other than it's non-negativity) are used.} \Cref{mainthm} follows simply by showing
that this is exactly the type of disturbance needed to dissolve a
logical pre- and post-selection paradox.

In the case of a finite state space $\Lambda$, the proof would run as
follows. By \cref{orthoprop}, when there is no intervening measurement,
the post-selection can occur. Hence there exists some $\lambda$
compatible with the preparation that makes the post-selection occur.
Consider $P$ with $\qabl{P} = 1$. If $\op(P|\lambda) = 0$ then
measurement-disturbance must always prevent the post-selection from
occurring, in contradiction with \cref{mdlem}. Hence outcome
determinism for sharp measurements gives $\op(P|\lambda) = 1$.
Repeating this for all other $P$ with $\qabl{P} = 1$, we find that
$\op(P|\lambda) = \qabl{P}$ for every $P \in \pset$. But since
$\op(P|\lambda)$ must satisfy the algebraic conditions, we have a
contradiction. We now present a formal proof that applies to an
arbitrary space of ontic states.

\begin{proof}[Proof of \cref{mainthm}] 
  The proof works by showing that in order to reproduce the ABL
  probabilities, there must exist ontic states $\lambda$ such that
  $\op (P|\lambda) = \qabl{P}$ for every $P \in \pset$.  Since $\op
  (P|\lambda)$ must satisfy the algebraic conditions on $\pset'$, and
  there is no extension of $\qabl{P}$ that does so, this is a
  contradiction, so no measurement noncontextual model is possible.

  It suffices to prove this for those $P \in \pset$ such that
  $\qabl{P} = 1$, since if $\qabl{P} = 0$ then $\qabl{I-P} = 1$ and,
  in the ontological model, we necessarily have $\op (P|\lambda) + \op
  (I-P|\lambda) = 1$.

  We start by reproducing the reasoning that led to \cref{ablrule} at
  the ontological level.  For concreteness, suppose that the
  post-selection works by making the projective measurement
  $\{\prj{\phi},I - \prj{\phi}\}$ and selecting the cases where the
  $\prj{\phi}$ outcome is obtained.  A projective measurement
  $\{P,I-P\}$ with L{\"u}ders-von Neumann update followed by a
  measurement of $\{\prj{\phi},I - \prj{\phi}\}$ is a method of
  measuring the POVM
  $\{E_{P,\phi},E_{P,\bar{\phi}},E_{\bar{P},\phi},E_{\bar{P},\bar{\phi}}\}$,
  where
  \begin{align}
    E_{P,\phi} & = P \prj{\phi} P & E_{P,\bar{\phi}} & = P (I -
    \prj{\phi})P \\
    E_{\bar{P},\phi} & = (I-P) \prj{\phi} (I-P) &
    E_{\bar{P},\bar{\phi}} & = (I-P) (I - \prj{\phi}) (I-P),
  \end{align}
  so we can calculate the joint probabilities in the ontological model
  as
  \begin{equation}
    \label{ontjoint}
    \qp (P,\phi|\psi) = \int_{\Lambda} \op (E_{P,\phi}|\lambda)d\mu(\lambda),
  \end{equation}
  and the marginal for passing the post-selection as
  \begin{equation}
    \label{ontmarj}
    \qp(\phi|\psi) = \int_{\Lambda}
    \op (E_{P,\phi}|\lambda)d\mu(\lambda) + \int_{\Lambda}
    \op(E_{\bar{P},\phi}|\lambda)d\mu(\lambda), 
  \end{equation}
  and so the probability for
  $P$ conditional on both pre- and post-selection is
  \begin{equation}
    \label{ontabl}
    \qabl{P} =
    \frac{\qp (P,\phi|\psi)}{\qp(\phi|\psi)} =
    \frac{\int_{\Lambda} \op (E_{P,\phi}|\lambda)d\mu(\lambda)}{\int_{\Lambda}
      \left(\op(E_{P,\phi}|\lambda) + \op(E_{\bar{P},\phi}|\lambda)\right)d\mu(\lambda)}.
  \end{equation}

  Now, if $\qabl{P} = 1$ then \cref{ontabl} implies
  \begin{equation}
    \label{ontint}
    \op (E_{P,\phi}|\lambda) = \op(E_{P,\phi}|\lambda) + \op(E_{\bar{P},\phi}|\lambda),
  \end{equation}
  on a set $\Omega_P$ such that $\mu(\Omega_P) = 1$. Let $\Lambda^{\phi} = \{\lambda \in \Lambda | \op
  (\prj{\phi}|\lambda) = 1\}$. We proceed by coarse-graining $\{E_{P,\phi},E_{P,\bar{\phi}},E_{\bar{P},\phi},E_{\bar{P},\bar{\phi}}\}$ in two different ways, applying \cref{coarse} each time.

  Firstly, since $E_{P,\phi} + E_{\bar{P},\phi} = \pvmc^\dagger(\prj{\phi})$, the RHS of \cref{ontint} equals $\op(\pvmc^\dagger(\prj{\phi})|\lambda)$, and therefore so does the LHS, $\op(E_{P,\phi}|\lambda)$. Hence given that $\op (\prj{\phi}|\lambda) = 1$ on $\Omega_P \cap
  \Lambda^{\phi}$, by \cref{mdlem}
  $\op(E_{P,\phi}|\lambda) = \op (\pvmc^{\dagger}(\prj{\phi})|\lambda) > 0$ on this set also.

  Secondly, $E_{P,\phi} + E_{P,\bar{\phi}} = P$ gives $\op (E_{P,\phi}|\lambda) +
  \op (E_{P,\bar{\phi}}|\lambda) = \op (P|\lambda)$ and thus $\op (P|\lambda) \geq
  \op (E_{P,\phi}|\lambda) > 0$ on $\Omega_P \cap \Lambda^{\phi}$.  By
  outcome determinism for sharp measurements, in fact $\op (P|\lambda)
  = 1$ on $\Omega_P \cap \Lambda^{\phi}$.

  Repeating this argument for every $P \in \pset$ such that $\qabl{P}
  = 1$, we have that, for every such $P$, there exists a set $\Omega_P
  \subseteq \Lambda$ such that $\mu(\Omega_P) = 1$ and $\op (P|\lambda) = 1$ on $\Omega_P \cap
  \Lambda^{\phi}$.

  Finally, notice that outside $\Lambda^\phi$, outcome determinism for sharp
  measurements gives $\op (\prj{\phi}|\lambda) = 0$.  By \cref{orthoprop}, $\ket{\psi}$ and
  $\ket{\phi}$ are nonorthogonal, which means that
  $\mu(\Lambda^{\phi}) > 0$:
  \begin{equation}
    0 < \abs{\braket{\phi|\psi}}^2 = \int_{\Lambda} \op (\prj{\phi}|\lambda) d\mu(\lambda) = \int_{\Lambda^{\phi}} d\mu(\lambda) = \mu(\Lambda^{\phi}).
  \end{equation}
  
  Now, $\Omega = \cap_{P \in \pset} \Omega_P$ is
  also measure one according to $\mu$ because it is the intersection
  of a finite number of measure one sets.  Thus $\mu(\Omega \cap
  \Lambda^{\phi}) = \mu(\Lambda^{\phi}) > 0$, in particular $\Omega \cap \Lambda^\phi$ is nonempty. Recall that on this set $\op (P|\lambda) = 1$ for $P \in \pset$ such that $\qabl{P} = 1$.  Since these
  assignments cannot be extended to $\pset'$ without violating the
  algebraic conditions, this contradicts the assumption that the model
  is non-contextual.
\end{proof}

\section{The connection to weak values}

Similar conclusions can be reached for an
alternative version of PPS paradoxes based on weak measurements.
Without going into details, an observable $A$ can be measured by
coupling the system to a continuous variable pointer system via a
Hamiltonian $H = g A \otimes p$, where $g$ is the coupling constant,
$A$ is the observable to be measured, and $p$ is the momentum of the
pointer.  If the parameters are chosen such that $gt \ll \Delta x$,
where $t$ is the duration of the measurement interaction and $\Delta
x$ is the initial position uncertainty of the pointer, then this is
called a ``weak measurement''.  

If the system is pre- and post-selected, with a weak measurement in
the middle, then, to first order in $gt$, the position distribution of
a suitably prepared pointer simply shifts by an amount
$gt w(A|\psi,\phi)$, where
\begin{equation}
  w(A|\psi,\phi) = \mathrm{Re} \left (
    \frac{\braket{\phi|A|\psi}}{\braket{\phi|\psi}} \right ),
\end{equation}
and $w(A|\psi,\phi)$ is called the \emph{weak value} of $A$.  Weak
values can lie outside the eigenvalue range of the operator $A$, in
which case they are called \emph{anomalous} weak values.

It is easy to check that the weak values assigned to a partial boolean
algebra of projection operators always satisfy the algebraic
\cref{ac2,ac3} with $f(P) = w(P|\psi,\phi)$.  This
is because, unlike the ABL probabilities, the denominator of
$w(P|\psi,\phi)$ does not depend on which projector we are measuring.
However, anomalous weak values mean that \cref{ac1} is sometimes
violated, i.e.\ $w(P|\psi,\phi)$ can be negative or greater than $1$.

It is also easy to verify that if $\qabl{P}$ is $0$ or $1$ then
$w(P|\psi,\phi) = \qabl{P}$ \cite{given}.  This means that, whenever there is a
logical PPS paradox for $\pset$, there is some projector in the
partial boolean algebra $\pset'$ that has an anomalous weak value.
This is because, by definition, there is no extension of the ABL
probabilities to $\pset'$ that satisfies all of the algebraic
conditions, and \cref{ac1} is the only one that can be violated by
weak values.  Therefore, logical PPS paradoxes will always show up as
anomalous weak values in the weak measurement version of the
experiment.

For example, in the three-box paradox, we have $w(\prj{1} +
\prj{2}|\psi,\phi) = 2$ and $w(\prj{3}|\psi,\phi) = -1$.

Because weak measurements do not disturb the state of the system to
first order in $gt$, strange behaviour of weak values is often thought
to be more puzzling than strange behaviour of ABL probabilities.
However, weak values should be interpreted with caution because they
are not probabilities, but rather small shifts in the distribution of
pointer position.  Nonetheless, it has recently been shown
\cite{weakctx} that anomalous weak values are proofs of contextuality,
Combined with our results above, this means 
logical PPS paradoxes are proofs of contextuality in both their strong
and weak measurement versions.

\section{Important features of the paradoxes}

\Cref{mainthm} establishes that logical PPS paradoxes are proofs of
contextuality.  However, classical analogues of violation of the
algebraic conditions have been reproduced by classical toy theories
\cite{kirk,lsmodel,maroney}, which do not appear to be contextual.  In
light of this, it is worth discussing the additional features of
logical PPS paradoxes that are essential to our proof, but do not
appear in the toy models.

\subsection{The importance of L{\"u}ders-von Neumann updates}

If we allow more general update rules for the intermediate
measurement, then we can obtain similar predictions to a logical PPS
paradox, but with orthogonal pre- and post-selection and without
contextuality.

For example, consider a qubit pre-selected in the state $\ket{0}$ and
post-selected in the state $\ket{1}$.  At an intermediate time, we
make a projective measurement $\{\prj{+},\prj{-}\}$, where $\ket{\pm}
\propto \ket{0} \pm \ket{1}$, in one of two different ways.

In the first method, upon obtaining outcome $\prj{+}$, we apply the
projection postulate as usual, but if the $\prj{-}$ outcome is
obtained then we reset the system to the $\ket{0}$ state.  This is a
valid state-update rule, as it corresponds to the quantum
instrument\footnote{Given a POVM $\{E_j\}$, a quantum instrument is a
  set of CP maps $\{\pvmc_j\}$ such that $\pvmc_j^{\dagger}(I) = E_j$
  and $\sum_j \pvmc_j$ is trace-preserving.  For any such instrument,
  it is possible to measure the POVM in such a way that the state
  update rule is $\rho \rightarrow \pvmc_j(\rho)/\tr(E_j \rho)$.  See
  \cite{heinosaariziman} for details.}
\begin{align}
  \pvmc_{+}(\rho) & = \prj{+} \rho \prj{+} & \pvmc_{-}(\rho) =
  \ket{0}\bra{-} \rho \ket{-}\bra{0}.
\end{align}
Clearly, if the post-selection succeeds, then the outcome of the
intermediate measurement must have been $\prj{+}$, since otherwise the
state of the system prior to post-selection would still be orthogonal
to $\ket{1}$, so we have $\qp
(\prj{+}|\ket{0},\{\pvmc_{+},\pvmc_{-}\},\ket{1}) = 1$.

In the second method, we do the opposite, applying the projection
postulate on obtaining the $\prj{-}$ outcome and resetting the the
system to the $\ket{0}$ state otherwise, which corresponds to the
instrument
\begin{align}
  \pvmc'_{+}(\rho) & = \ket{0}\bra{+} \rho \ket{+}\bra{0} & \pvmc'_{-}(\rho) =
  \prj{-} \rho \prj{-}.
\end{align}
By the same reasoning, we can conclude that $\qp
(\prj{-}|\ket{0},\{\pvmc'_{+},\pvmc'_{-}\},\ket{1}) = 1$.

If we allow ourselves to combine the probabilities for different
intermediate measurements in the same way that we did for logical PPS
paradoxes, setting $f(\prj{+}) = \qp
(\prj{+}|\ket{0},\{\pvmc_{+},\pvmc_{-}\},\ket{1})$ and $f(\prj{-}) =
\qp (\prj{-}|\ket{0},\{\pvmc'_{+},\pvmc'_{-}\},\ket{1})$, then this
would violate the algebraic conditions.  However, this is not a proof
of contextuality as it can be easily accounted for by
measurement-disturbance in a non-contextual ontological model.

Specifically, it occurs in a suitably modified version of Spekkens'
toy theory \cite{toytheory} in which we modify the
measurement-disturbance slightly in order to model the ``resetting''
that occurs for one of the outcomes.

Briefly, the Spekkens' toy bit has four ontic states, which we label
$1$, $2$, $3$, $4$.  The $\ket{0}$ state is modelled by a uniform
distribution over $1$ and $2$, the $\ket{1}$ state by a uniform
distribution over $3$ and $4$, the $\ket{+}$ state by a uniform
distribution over $1$ and $3$, and the $\ket{-}$ state by a uniform
distribution over $2$ and $4$.  The post-selection consists of
checking whether the ontic state is $3$ or $4$ and rejecting if it is
not.  

For a projective measurement of $\{\prj{+},\prj{-}\}$ with
L{\"u}ders-von Neumann update, we output $\prj{+}$ if the ontic state
is $1$ or $3$ and then disturb the system by doing nothing with
probability $1/2$ and swapping $1$ and $3$ with probability $1/2$, and
we output $\prj{-}$ if the ontic state is $2$ or $4$ and then disturb
the system by doing nothing with probability $1/2$ and swapping $2$
and $4$ with probability $1/2$.

For the modified update rule $\{\pvmc_{+},\pvmc_{-}\}$, the only
thing we change is that, upon obtaining the $\prj{-}$ outcome, we swap
$1$ and $2$ instead of $2$ and $4$.  Similarly for
$\{\pvmc'_{+},\pvmc'_{-}\}$, upon obtaining the $\prj{+}$ outcome, we
swap $1$ and $2$ instead of $1$ and $3$.  It is easy to see that this
setup predicts the same probabilities as quantum theory.  Indeed, when
the post-selection succeeds, only the $\prj{+}$ outcome can occur when
the instrument $\{\pvmc_{+},\pvmc_{-}\}$ is used, as this is the only
way the system can end up in $3$ or $4$ if it starts out in $1$ or
$2$, and similarly only the $\prj{-}$ outcome can occur when the
instrument $\{\pvmc'_{+},\pvmc'_{-}\}$ is used.

The model just described is very similar to the toy model for logical
PPS paradoxes introduced in \cite{lsmodel}, with the exception that,
instead of modifying the measurement-disturbance, the model of
\cite{lsmodel} eliminates it for the intermediate measurement outcome
that is supposed to have probability $0$.

What we learn from this is that, in order to imply contextuality, it
is not enough to just have a set of predictions for intermediate
measurements that violate the algebraic conditions.  It is important
that the intermediate measurements have L{\"u}ders-von Neumann update,
because this allows us to infer that the pre- and post-selection must
be nonorthogonal, and that the intermediate measurement cannot make
the post-selection go from being possible to being impossible.  Both
of these were needed for the proof of \cref{mainthm}.  This explains
why, although violations of the algebraic conditions have been found
in various toy models, there are no true logical PPS paradoxes in
noncontextual theories when we attempt to faithfully model
L{\"u}ders-von-Neumann updates.

\subsection{The importance of $0$/$1$ probabilities}

Another important aspect of logical PPS paradoxes is that we demanded
that the probabilities of the intermediate measurement outcomes should
all be $0$ or $1$.  Dropping this requirement can yield probabilities
that violate the algebraic conditions, but nonetheless still have a
noncontextual model.

An example of this is the ``quantum cheshire cat'' \cite{chesire}.  In
this experiment, we have a two qubit system which is pre-selected in
the state
\begin{equation}
  \ket{\psi} = \frac{1}{\sqrt{2}} \left ( \ket{0} + \ket{1} \right )
  \otimes \ket{0},
\end{equation}
and post-selected in the state
\begin{equation}
  \ket{\phi} = \frac{1}{\sqrt{2}} \left ( \ket{0} \otimes \ket{0} +
    \ket{1} \otimes \ket{1} \right ).
\end{equation}

Now, the ABL probabilities satisfy $\qabl{\prj{1} \otimes I} = 0$, but
$\qabl{\prj{1} \otimes \prj{+}} = \qabl{\prj{1} \otimes \prj{-}} =
1/6$.  This is a violation of the algebraic conditions because
$\prj{1} \otimes I = \prj{1} \otimes \prj{+} + \prj{1} \otimes
\prj{-}$, so we should have $f(\prj{1} \otimes I) = f(\prj{1} \otimes
\prj{+}) + f(\prj{1} \otimes \prj{-})$.  However, all of the states
and measurements in this experiment are correctly modelled by
Spekkens' toy theory so, without going into detail, we can conclude
that this does not establish contextuality.  Thus, the condition that
logical PPS paradoxes should involve $0/1$ probabilities is
essential.  A violation of the algebraic conditions on its own is not
enough to establish contextuality.

Interestingly though, the weak measurement version of the quantum
cheshire cat does establish contextuality, because it involves
anomalous weak values.  Since an ABL probability of $0$ implies a weak
value of zero, we have $w(\prj{1} \otimes I|\psi,\phi) = 0$, but we
also find that $w(\prj{1} \otimes \prj{+}|\psi,\phi) = 1/2$ and
$w(\prj{1} \otimes \prj{-}|\psi,\phi) = -1/2$.  These satisfy the
condition $w(\prj{1} \otimes I|\psi,\phi) = w(\prj{1} \otimes
\prj{+}|\psi,\phi) + w(\prj{1} \otimes \prj{-}|\psi,\phi)$, as weak
values always do, but the value $-1/2$ lies outside the eigenvalue
range of the projector $\prj{1} \otimes \prj{-}$, so it is anomalous.
It is rather intriguing that an experiment that can be modelled
non-contextually in its strong measurement version can nonetheless
become contextual when the measurements are weakened.  It would be
interesting to know if violations of the algebraic conditions for
non-$0$/$1$ ABL probabilities always imply anomalous weak values in
this way.

\section{Conclusion}

In conclusion, we outline what \cref{mainthm} tells us about logical
pre- and post-selection paradoxes. We follow the three-pronged
approach of the conclusions in \cite{weakctx}, where broadly similar
techniques were used to show that anomalous weak values are also
proofs of contextuality.

Firstly, the proof enables a classification of possible
interpretations of a logical pre- and post-selection paradox. Suppose
that, despite \cref{mainthm}, we demanded an ontological model for,
say, the three-box paradox. Then at least one of the requirements of
non-contextuality must be violated. It could be the algebraic
conditions, i.e.\ the ball really is in two boxes at once. But it
could instead be the outcome determinism of sharp measurements, i.e.\
there could be no fact about which box the ball is in until the
measurement. Finally it could be that the intermediate measurement
\emph{always} disturbs the post-selection, in violation of the
measurement non-contextuality of the post-selection. Since we view any
form of contextuality as a deficiency in the explanation offered by an
ontological model, we see no particular reason to privilege one of
these possibilities over the others. A sensible option is to reject
the ontological models framework entirely, but without a replacement
it is impossible to say anything rigorous about what lies behind these
paradoxes.

Secondly, the proof suggests that several aspects of these paradoxes
are crucial to preventing a compelling classical explanation, despite
having received little attention thus far. If the intermediate
measurements were not projective, then the pre- and post-selected
states need not overlap, and then there would be no
reason to think that the post-selection could occur in the absence of
a disturbance. There would also be no reason to think that the
intermediate measurements were reading out a pre-defined value. If the
state update rule for the intermediate measurement was something other
than L\"{u}ders-von Neumann rule, or \cref{idpart} was not a feature
of quantum theory, then there would be no reason to think that the
intermediate measurement sometimes has no effect on the
post-selection.

Finally, the proof helps to identify the issues that would have to be
addressed in order to turn a logical pre- and post-selection paradox
into an experimental proposal for demonstrating non-classicality
(i.e.\ a proposal that doesn't render the experiment redundant by
simply assuming all of quantum theory \emph{a priori}). For example,
one would first need an experimental version of the argument from
preparation non-contextuality to outcome determinism for sharp
measurements, for example by using the predictability of the
intermediate measurements on states that overlap with the original
preparation \cite{ravi}. One would need an analysis of how close to
the unrealistic $0$ and $1$ probabilities of the pre- and
post-selected values the experiment would have to come. Finally one
would need an operational method of testing that the post-selection
measurement is the same (to some appropriate level of approximation)
whether preceded by an intermediate projective measurement or the
mixture of channels in \cref{idpart}.  

\printbibliography
\end{document}